\documentclass[journal]{IEEEtran}
\IEEEoverridecommandlockouts
\usepackage[noadjust]{cite}
\usepackage{amsmath,amssymb,amsfonts}
\usepackage{amsthm}
\usepackage{algorithmic}
\usepackage{relsize}
\usepackage{algorithm}
\usepackage{graphicx}
\usepackage{subcaption}
\usepackage{comment}
\usepackage{textcomp}
\usepackage[T1]{fontenc}
 \usepackage{url}
 \newtheorem{assumption}{Assumption}
  \newtheorem{Remark}{Remark}
    \newtheorem{Theorem}{Theorem}
     \newtheorem{Corollary}{Corollary}
     
 \usepackage{xcolor}
\usepackage{tikz}
\def\BibTeX{{\rm B\kern-.05em{\sc i\kern-.025em b}\kern-.08em
    T\kern-.1667em\lower.7ex\hbox{E}\kern-.125emX}}
\usepackage{enumerate}
\usepackage{paralist} 
\usepackage{afterpage}
\usepackage{amsbsy}

\usepackage{mathtools}
\mathtoolsset{showonlyrefs}
\usepackage{placeins}

\allowdisplaybreaks
\begin{document}
	
	\title{Distributed detection of ARMA signals}
	\author{Jo\~ao Domingos and Jo\~ao Xavier
		\thanks{\hspace{-.61cm} Jo\~ao Domingos and Jo\~ao Xavier {\tt\small oliveira.domingos@ist.utl.pt, jxavier@isr.tecnico.ulisboa.pt.}
			are with the Instituto Sistemas e Rob\'otica, Instituto Superior T\'ecnico, Universidade de Lisboa, Portugal.}
	}
	%
	\newtheorem{theorem}{Theorem}
	\newtheorem{lemma}{Lemma}
	\newtheorem{remark}{Remark}
	\newtheorem{example}{Example}
	\newtheorem{result}{Result}
	\newtheorem{corollary}{Corollary}[theorem]
	\maketitle

\begin{abstract}
This paper considers a distributed detection setup where agents in a network want to  detect a time-varying signal  embedded in temporally correlated  noise. In contrast to most previous work, which considers only a constant signal embedded in temporally uncorrelated noise, in this paper both signal and noise can have a general auto-regressive moving average (ARMA) temporal structure. The signal of interest is the impulse response of an ARMA filter, and the noise is the output of yet another ARMA filter which is fed white Gaussian noise. 

For this extended problem setup, which can prompt novel behaviour, we propose a comprehensive solution.
First, we extend the well-known running consensus detector (RCD) to this correlated setup; then, we design an efficient implementation of the RCD by exploiting the underlying ARMA structures; and, finally, we  derive the theoretical asymptotic performance of the RCD in this ARMA setup. It turns out that
the error probability at each agent exhibits one of two regimes: either (a) the error probability decays exponentially fast to zero or  (b) it converges to a strictly positive error floor. While regime (a) spans staple results in large deviation theory, regime (b) is new in distributed detection and is elicited by the  ARMA setup. 
We fully characterize these two scenarios: we give necessary and sufficient conditions, phrased in terms of the zero and poles of the underlying ARMA models, for the emergence of each regime,  and  provide closed-form expressions for both the decay rates of regime (a) and the positive error floors of regime (b). Our analysis also shows that the ARMA setup leads to two novel features: (1) the threshold level used in RCD can influence the asymptotics of the error probabilities and (2) some agents might be \emph{weakly informative}, in the sense that their observations do not improve the asymptotic performance of RCD and, as such, can be safely muted to save sensing  resources. Numerical simulations illustrate and confirm the theoretical findings.
\end{abstract}

\begin{IEEEkeywords}
Distributed Detection, Binary Hypothesis Testing, ARMA, Temporal Correlation, Running Consensus Detector, Cholesky Decomposition.
\end{IEEEkeywords}

\section{Introduction}
\label{sec:introduction}
Binary detection addresses the problem of deciding which statistical source, out of two possible ones, generated some observed data. 
Detection can be carried out in centralized or distributed setups, depending mainly on how the data is observed. 
In centralized setups, a single agent observes the whole data and can address the decision problem by itself. In distributed setups the data is globally observed by multiple agents, but each agent observes only a local portion of it. Because each agent now observes only part of the global data, no single agent can decide by itself in an optimal way; agents are better off addressing the decision problem cooperatively. 

\subsection{Distributed setup}
This paper considers detection in a distributed setup, in which the global data is a discrete-time, vector-valued signal of dimension $N$. Of such stream of global data vectors, each agent~$i$ (with $1 \leq i \leq N)$ observes locally only the component $i$ at each time step. The  agents want to decide between two hypotheses: 
\begin{itemize}
    \item hypothesis $H_1$ (signal+noise): the global data is a given deterministic signal embedded in additive noise;
    
    \item hypothesis $H_0$ (noise only): the global data is just noise,
\end{itemize}
where both the signal and the noise have  pre-specified---but otherwise generic---auto-regressive moving average (ARMA) temporal structures. 

 

\subsection{Related work in distributed setups}

To better situate this paper in the  current landscape of detection for distributed setups, a review of closest work follows.

Broadly speaking, 
detection in distributed setups splits in two main groups, fusion center approaches~\cite{Chi_squared_approximation_2022,change_detection_FC_2022,Limit_Memoryless_Quant_2022,Entropy_Based_Quantization_2022,Narrow_Wide_sensing_2021,Multiplicative_fadding_2021,sparse_signal_detection_2022} and fully decentralized approaches~\cite{Byzantine_Sensors_2021,consensus_sequential_HP_2021,Heterogeneous_networks_2014,Dist_Quickest_Detection_2018,Joint_detection_estimation_2020,Nonlinear_models_2017,non_idd_product_2018,Phase_Transition_for_Communication_Noise_2012,Running_consensus_w_random_networks_2011,Distributed_MAP_2022,Dist_ML_Detection_2018,Kalman_fitler_detection,Least_favourable_densities_2017,Arith_Geo_Fusion_2022,marginals_not_joint_2021,Tutorial_Dist_Learning_2016,activ_HP_2017,robust_test_multi_hypothesis_2018,CISPRT_2016,Improved_CISPRT_2019,Gaussian_shift_in_mean_2015,robust_seq_testing_2018,Distributed_Quantized_2020}:
\begin{itemize}
    \item In fusion center approaches, agents communicate to a central agent some test statistics (often quantized) of their local data. The central agent, upon receiving these statistics,  runs them  through a fusion rule to arrive at the collective decision~\cite{Chi_squared_approximation_2022,change_detection_FC_2022,Limit_Memoryless_Quant_2022,Entropy_Based_Quantization_2022,Narrow_Wide_sensing_2021,Multiplicative_fadding_2021,sparse_signal_detection_2022};
  \item In decentralized approaches,
  such hierarchical architecture is removed by flattening it: no special central agent exists; all agents are on an equal footing and communicate bidirectionally with available neighbor agents to arrive cooperatively at the decision.
  At each time step, each agent updates its local decision by assimilating its newly observed (private) data and by blending the revised local decision with those of its neighbors.
  
   Part of the work on decentralized detection focuses on sequential detection, where observations are collected only until a decision that surpasses acceptable thresholds on its accuracy can be made~\cite{Byzantine_Sensors_2021,consensus_sequential_HP_2021,Heterogeneous_networks_2014,non_idd_product_2018,CISPRT_2016,Improved_CISPRT_2019,Gaussian_shift_in_mean_2015,robust_detect_by_robust_est_2017,Distributed_Quantized_2020,Nonlinear_models_2017}. Other part of the work focuses on non-sequential detection, with composite hypotheses~\cite{Arith_Geo_Fusion_2022,marginals_not_joint_2021,Tutorial_Dist_Learning_2016,activ_HP_2017}.
  \end{itemize}
  ~\\
This paper belongs to the second group of approaches---the fully decentralized ones.
We consider a decentralized, binary detection problem where the signal being detected is known but time-varying, and the observation noise is temporally colored Gaussian noise. Our method uses the running consensus detector (RCD) proposed in~\cite{running_consensus_original_2009} and further considered in~\cite{Running_consensus_w_random_networks_2011}. Despite the rich literature on decentralized detection we found no other work that considers distributed detection with temporally correlated observation noise. The most adjacent work is~\cite{Kalman_fitler_detection}, where Pang et al. consider a decentralized, binary detection problem with temporally correlated (yet random) detection signal; specifically, they study the detection of a partially observed Markov process with i.i.d. white Gaussian noise. 
 
 In this paper we assume that the detection signal is the impulse response of an ARMA model and the observation noise is the output of  another ARMA filter which is fed  white Gaussian noise. ARMA signals have been well studied in the signal processing literature because they are temporal equivalents of rational $\mathcal{Z}$ transforms~\cite{Oppenheim}. Because a vast literature exists on how to approximate a desired frequency response with ARMA models~\cite{digital_filter_design}, the ARMA assumption should not be too restrictive if one can accurately approximate the frequency response of the detection signal and the power spectral density of the noise (assuming wide sense stationarity). One could also employ an alternative approach where the involved parameters are estimated online: in~\cite{known_signal_unknown_gaussian_noise_2} Kay tackles the problem of detecting a known (not ARMA) signal embedded in autoregressive (AR) noise with unknown parameters; it is shown that, asymptotically, the generalized likelihood ratio test (GLRT) is able to match the optimal performance of the log-likelihood ratio (LLR) detector by iteratively updating the autoregressive parameters using maximum likelihood estimates. This work is further generalized in~\cite{random_signal_unknown_gaussian_noise_1}, where the detection signal may also be unknown. 
 
 Let us note, however, that the aforementioned deals with AR processes only and with fully centralized architectures.
While~\cite{known_signal_unknown_gaussian_noise_2,random_signal_unknown_gaussian_noise_1} present tractable centralized detectors when several parameters are unknown, we study the implementation and asymptotic performance of a distributed detector for a known ARMA signal with a known additive ARMA (not necessarily restricited to be an AR) noise. 

\subsection{Contributions}
To the best of our knowledge, this is the first work that considers distributed detection with temporally correlated observation noise. In more detail, we claim three concrete contributions:
\begin{enumerate}
	\item By exploring the structure of ARMA filters, we show that the running consensus detector (RCD) can be efficiently implemented with computational and memory requirements independent of the number of observations;
	\item More importantly, we fully characterize the asymptotic performance of RCD, proving 
	that the probability of false alarm and miss detection can exhibit two asymptotic regimes: (a) the probabilities decay exponentially fast to zero, or (b) they converge to strictly positive error floors. Regime (b) is novel in the distributed context and arises from the generic ARMA setup. We 
	clarify which regime is triggered on the basis of the zeros and poles of the ARMA models at play, and derive closed-form expressions
	for both the exponential decay rates and the error floors; such expressions reveal that the threshold constant used in RCD is irrelevant in regime (a) but highly relevant for the limiting error floors of regime (b).
	
	\item We show that the ARMA setup can also induce \emph{weakly informative} agents. These are agents whose data cannot improve the asymptotic performance of the RCD. Their data can thus be discarded without hurting this asymptotic metric. Such property might be explored in network compression schemes where RCD could run in a (possibly much smaller) connected network obtained by pruning most (if not all)  of the weakly informative agents.
\end{enumerate}

\subsection{Structure of the paper}

The remaining of the paper is organized as follows. Section~\ref{sec:Preliminaries_notation} sets the notation and presents some necessary background on ARMA models and $\mathcal{Z}$-transforms. The distributed detection problem is formalized in section~\ref{sec:problem_formulation}, where we introduce the sensing and communication models. Section~\ref{sec:running_consensus} extends the generic running-consensus detector (RCD) for the ARMA case. Section~\ref{sec:Asymptotic_Detection_Errors} derives the distribution of the additive colored noise and provides an efficient implementation of RCD by using techniques of ARMA filtering. In section~\ref{sec:theory_on_RCD} we analyze the asymptotic performance of RCD. 
Section~\ref{sec:numerical_section} finishes the paper by considering the problem  of detecting a DC level embedded in several classes of colored noise.
\section{Preliminaries: $\mathcal{Z}$ transforms}
\label{sec:Preliminaries_notation}
For future reference, we set notation and review some key properties of rational $\mathcal{Z}$ transforms~\cite{Oppenheim}.
\subsection{Notation}
The index $i=1,\dots,N$ identifies agents in a network, and the index $k=0,1,2,\dots$  represents discrete-time. 

~\\[5pt]
\noindent\textbf{Vectors and matrices.} A plain vector notation such as $z(k)$ represents a \emph{snapshot} of the network at time $k$: $$z(k) = \begin{bmatrix} z_1(k) \\ \vdots \\ z_N(k) \end{bmatrix} \in \mathbf{R}^N,$$
with $z_i(k) \in \mathbf{R}$ being associated with agent~$i$.
A boldface vector notation  $\pmb{z}(k)$ represents the \emph{history} of the network from $0$ until time $k$:
$$\pmb{z}(k) = \begin{bmatrix} \pmb{z}_1(k) \\ \vdots \\ \pmb{z}_N(k) \end{bmatrix} \in \mathbf{R}^{N(k+1)},$$
with $\pmb{z}_i(k) \in \mathbf{R}^{k+1}$ containing the history of agent~$i$,
$$\pmb{z}_i(k) = 
\begin{bmatrix} z_i(0) \\ z_i(1) \\ \vdots \\ z_i(k) \end{bmatrix} \in \mathbf{R}^{k+1}.
$$

Note that the plain $z(k)$ has a fixed dimension (regardless of $k$), whereas the boldface $\pmb{z}_i(k)$ and $\pmb{z}(k)$ have a dimension that grows with $k$.
This boldface principle carries over to matrices.
For example, $\pmb{Z}(k)$ represents an history of the network from $0$ till time $k$, $$\pmb{Z}(k)=\text{diag}\{\pmb{Z}_1(k),\dots,\pmb{Z}_N(k)\}\in \mathbf{R}^{N(k+1)\times N(k+1)},$$
with $\text{diag}\{ \cdot \}$ denoting diagonal stacking of matrices and $\pmb{Z}_i(k) \in \mathbf{R}^{(k+1)\times (k+1)}$ containing the history associated with  agent~$i$. 


~\\[5pt]
\noindent\textbf{Assorted symbols.}
Symbols $\pmb{1}_n$ and $I_n$ denote the $n$- dimensional vector with all components equal to one and the $n \times n$ identity matrix, respectively,  with the subscript $n$ dropped when easily inferred from the context. The Euclidean norm of vectors is $\| \cdot \|$. 
For a $m\times n$ matrix $X$, written in columns as $\begin{bmatrix} x_1 & \cdots & x_n \end{bmatrix}$, the symbol $\text{vec}(X)$ denotes the $mn$-dimensional column vector obtained by stacking the columns of $X$ from left to right: $\text{vec}(X) = \begin{bmatrix} x_1^T & \cdots & x_n^T \end{bmatrix}^T$.

The notation $x \sim  \mathcal{N}( \mu, \Sigma )$ means that the random vector $x$ follows a Gaussian distribution with mean vector $\mu$ and covariance matrix $\Sigma$. The symbol $\mathbb{E}( x )$ denotes the expectation operator applied to a random vector $x$, and $\mathbb{P}( A )$ denotes the probability of an event $A$.

Sequences are often compared through the
big-oh notation $\mathcal{O}( \cdot )$ and the
little-oh notation
$o( \cdot )$. For sequences $\left( f(k) \right)_{k \geq 0}$ and $\left( g(k) \right)_{k \geq 0}$, the notation  $f(k)=\mathcal{O}(g(k))$ means that there exists a constant $M>0$ and an order $k_0$ such that $k \geq k_0$ implies $| f(k) | \leq M | g(k) |$. The notation $f(k)=o(g(k))$ means that the previous constant $M$ can be as small as desired, that is, for any $M > 0$ there is an order $k_0$ such that $k \geq k_0$ implies $| f(k) | \leq M | g(k) |$ (when $g(k) \neq 0$ for $k$ sufficiently large, the notation $f(k)=o(g(k))$ just indicates that $f(k) / g(k) \rightarrow 0$ as $k \rightarrow \infty$).

\subsection{$\mathcal{Z}$ transforms}

 The $\mathcal{Z}$-transform of a 
 sequence $( x(k) )_{k \geq 0}$ is the complex-valued function
\begin{align}
\label{eqn:def_Z_transform}
X(z^{-1})= \sum_{k \geq 0} x(k)\,z^{-k},\quad z\in \mathcal{ROC}\subseteq \mathbf{C},
\end{align}
where $\mathcal{ROC}:=\{z \in \mathbf{C} \colon  \sum_{k \geq 0} |x(k)\,z^{-k}|< \infty \}$ denotes its region of convergence (ROC). In this paper we consider only $\mathcal{Z}$-transforms  that are rational,
\begin{align}
\label{eqn:rational_Z}
X(z^{-1})&= A\, \dfrac{ \prod\limits_{j=1}^q(1-z_j\,z^{-1})}{ \prod\limits_{j=1}^p(1-p_j\,z^{-1})}, \quad |z|>\rho,
\end{align}
with $A$ being a constant, 
the roots $\{z_j\}_{j=1}^q$ of the numerator polynomial being the zeros of $X(z^{-1})$, and  the roots $\{p_j\}_{j=1}^p$ of the denominator polynomial being the poles of  $X(z^{-1})$; in this case, $\mathcal{ROC}:=\{z \colon | z | > \rho \}$, with $\rho := \text{max}\{ | p_j | \colon j \}$.

~\\[5pt]
\textbf{Partial fraction expansion.} 
It is well known that the sequence $\left( x(k) \right)_{k \geq 0}$ can be retrieved back from its $\mathcal{Z}$-transform. For rational $\mathcal{Z}$-transforms such retrieval can be carried out by the standard technique of partial fraction expansion: assuming that $X(z^{-1})$ in~\eqref{eqn:rational_Z} is strictly proper (orders $p,q$ satisfy $p > q$) and has only simple poles ($p_m \neq p_n$ for $m \neq n$), first represent $X( z^{-1} )$ as
  \begin{align}
  \label{eqn:proper_partial_fraction_general}
 X(z^{-1})&=\sum_{j=1}^p \frac{r_j}{1-p_j\,z^{-1}},
 \end{align}
 with residuals $r_j$ given by
 \begin{align}
 \label{eqn:residual_formula_generic}
 r_j &=(1-p_j\,z^{-1})\,X(z^{-1})\Big|_{z=p_j};
 \end{align}
 now, invert each of the $p$ simple fractions to get
 \begin{align}
 x(k) = \sum_{j=1}^{p}\,r_j\,p_j^k, 
 \label{eqn:partial_fraction_decomposition}
 \end{align}
 which shows that the sequence $\left( x(k) \right)_{k \geq 0}$ is a sum of $p$ discrete-time complex exponential signals $p_j^k$.

~\\[5pt]
\textbf{ARMA filters.} Consider an  autoregressive moving average (ARMA) $(p,q)$ filter with  coefficients $\{b(j)\}_{j=1}^p,\{a(j)\}_{j=1}^q$, and let $\left( x(k) \right)_{k \geq 0}$ be an input signal, with $\left( y(k) \right)_{k \geq 0}$ being the corresponding output signal. The input and output signals are related in the time-domain by the linear recursion
\begin{align}
y(k)=\sum\limits_{j=1}^p b(j)\,y(k-j)+x(k)+\sum\limits_{j=1}^q a(j)\,x(k-j),
\label{eqn:ARMA_def}
\end{align}
(with the understanding that $x(k) = 0$ for $k < 0$)
and in the $\mathcal{Z}$-domain by 
\begin{align}
\label{eqn:rational_Z_H}
Y(z^{-1})&=\dfrac{1- \sum\limits_{j=1}^q b(j)\,z^{-j}}{1 + \sum\limits_{j=1}^p a(j)\,z^{-j}} \, X(z^{-1}).
\end{align}
The nomenclature autoregressive moving average appears since, in~\eqref{eqn:ARMA_def}, we have $p$ autoregressive terms of $y(k)$ and $q$ moving average terms of $x(k)$. Equations~\eqref{eqn:ARMA_def} and~\eqref{eqn:rational_Z_H} are equivalent representations of an ARMA$(p,q)$ filter, the  time representation~\eqref{eqn:ARMA_def} being more useful for implementation, and the $\mathcal{Z}$ representation~\eqref{eqn:rational_Z_H} being more useful for analysis.

\section{Problem Formulation}
\label{sec:problem_formulation}
Assume $N$ agents collect local observations, 
whose statistics are shaped differently by two possible states of the world, labelled $H_1$ or $H_0$.
We address the problem of designing a collaborative algorithm that advances through local interactions between agents and leads them to collectively decide which hypothesis, $H_1$ or $H_0$, is in force.

\subsection{Sensing and communication models} 
Specifically, under hypothesis $H_1$, each agent~$i$ observes a noisy version of a known deterministic signal $\theta_i(k)$, while, under hypothesis $H_0$, each agent~$i$ just observes noise:
\begin{align}
H_1&:\enspace y_i(k)=\theta_{i}(k)+n_i(k),\enspace i=1,\dots,N \label{eqn:hypothesis_formulation} \\
H_0&:\enspace y_i(k)=n_i(k),\enspace i=1,\dots,N,
\end{align}
for $k \geq 0$ and $1 \leq i \leq N$.

We assume that both the noise $\left( n_i(k) \right)$
and the  signal $\left( \theta_i(k) \right)$ follow ARMA structures, in the sense that they are modeled as outputs of ARMA filters. For the  noise model, the input of the corresponding filter is $\sigma \, \epsilon_i(k) $, where $\sigma > 0$ and $\left( \epsilon_i(k) \right)$ is white Gaussian noise (WGN) (so,  $\sigma\,\epsilon_i(0),\dots,\sigma\,\epsilon_i(k)\stackrel{i.i.d}{\sim} \mathcal{N}(0,\sigma^2)$); for the signal model, the input of the corresponding filter is $A \, \delta(k)$, where $A \neq 0$ and $\delta$ denotes the unit impulse ($\delta( 0 ) = 1$ and $\delta( k ) = 0$ for $k > 0$). In view of~\eqref{eqn:ARMA_def}, this gives
\begin{align} 
\label{eqn:noise_ARMA}
n_i(k)&=\sum_{j=1}^{p_i} a_{i}(j)\,n_i(k-j) + \sigma\, \epsilon_i(k)+\sigma\,\sum_{j=1}^{q_i} b_{i}(j)\, \epsilon_i(k-j)\,
\\ 
\label{eqn:signal_ARMA} \theta_i(k)&=\sum_{j=1}^{\bar{p}_i}\bar {a}_{i}(j)\,  \theta_i(k-j) + A\,\delta(k)+A\,\sum_{j=1}^{\bar{q}_i}\bar {b}_{i}(j)\, \delta(k-j),
\end{align}
for known $A$, $\sigma$ and filter coefficients $a_i(j), \overline a_i(j), b_i(j), \overline b_i(j)$.

~\\[5pt]
\textbf{Communication model.} 
Besides sensing, agents can also communicate with neighbors by exchanging information over a network. The network is modelled as an undirected graph $G=(V,E)$, where $V=\{1,\dots,N\}$ is the set of agents, and $E$ is the set of communication channels. Edge $(i,j)\in E$ indicates that agents $i$ and $j$ can  exchange information bidirectionally between them. We let $\mathcal{N}_i$ denote the neighbours of agent~$i$: $\mathcal{N}_i:=\{j\in V: \,\,(i,j)\in E  \}$.

Each agent $i$ must decide between hypothesis $H_0$ and $H_1$ by processing its observations $y_i(k)$ sequentially in time $k$ while using the communication network to receive some local decision variables $l_j(k)$ from neighbouring nodes $j\in \mathcal{N}_i$ (how each agent $i$ computes its local variable $l_i(k)$ is detailed in 
Section~\ref{sec:running_consensus}). We assume that the local measurements $y_i(k)$ and the broadcasting  happen synchronously. So, in each time step $k$, agent $i$ simultaneously observes $y_i(k)$, broadcast its local decision variable $l_i(k)$  to neighbours, and receives neighbouring variables $l_j(k)$ for $j\in \mathcal{N}_i$.

\subsection{Assumptions} 
\label{sec:assumptions}
We work under the following three assumptions.
\begin{assumption} (Noise is spatially independent)
	\label{assumption:spatial_uncorrelatedness}
	For each time instant $k$, the additive noise $n_i(k)$ is independent across agents, meaning that the vector $n(k)=\{n_i(k)\}_{i=1}^N$ has independent components. Note that, however, each component $\left( n_i(k) \right)_{k \geq 0}$ can be temporally correlated;
\end{assumption}
\begin{assumption} (Network is connected)
	\label{assumption:connected}
	The communication network $G$ is a connected graph, that is, there exists a path (possibly, with multiple hops) joining any two agents. 
\end{assumption}
This is a common (and necessary) assumption, for, otherwise, the observations of isolated agents could not be taken into account. For further reference, we associate with $G$ a pre-chosen symmetric \emph{weight} matrix $W \in \mathbf{R}^{N\times N}$ that conforms to the sparsity of the network ($W_{i,j} = 0$ if and only if $j\notin \mathcal{N}_i$) and enjoys an eigenvalue decomposition of the form
	\begin{align}
	W= \underbrace{\begin{pmatrix}
	\tfrac{1}{\sqrt{N}}\mathbf{1}_N &  {S}
	\end{pmatrix}}_{Q_W} \underbrace{\begin{pmatrix}
	1 &  \\
	& D
	\end{pmatrix}}_{\Lambda_W} \underbrace{\begin{pmatrix}
	\tfrac{1}{\sqrt{N}}\mathbf{1}_N^T\\
	{S}^T
	\end{pmatrix}}_{Q_W^T}, \label{eqn:mW}
	\end{align}
	where $D \in \mathbf{R}^{(N-1)\times (N-1)}$ is a diagonal matrix with diagonal entries $d_n$ satisfying
	$| d_n | < 1$, for $1 \leq n \leq N-1$, and $S \in \mathbf{R}^{N \times (N-1)}$ has orthonormal columns ($S^T S = I$), each column orthogonal to $\pmb{1}_N$, that is $S^T \pmb{1}_N = 0$.
 
Note that $W$ is therefore both row-stochastic ($W \mathbf{1}_N = \mathbf{1}_N$) and column-stochastic ($\pmb{1}_N^T W = \pmb{1}_N^T$). It is well known that such matrices exist and can even be obtained in a distributed way, for instance, when $W$ is filled with Metropolis weights.
\begin{assumption} (ARMA Models)
	\label{assumption:ARMA_model}
	For each agent $i$, the rational  $\mathcal{Z}$-transform
	\begin{align}
 \label{eqn:equations_freq_domain_signal_assumptions}
	 \dfrac{A}{\sigma}\frac{1-\sum\limits_{j=1}^{q_i} a_{i}(j)\,z^{-j}}{1+\sum\limits_{j=1}^{p_i} b_{i}(j)\,z^{-j}}\, \frac{1+\sum\limits_{j=1}^{\bar{q}_i} \bar{b}_{i}(j)\,z^{-j}}{1-\sum\limits_{j=1}^{\bar{p}_i} \bar{a}_{i}(j)\,z^{-j}}
	\nonumber \\
	\end{align}
 is assumed to have no zero-pole cancellations (this entails no loss of generality). Let $Q_i = q_i + \overline q_i$ and $P_i = p_i + \overline p_i$ be its number of zeros and poles, respectively, with the $P_i$ poles denoted by $\{ p_{i,1}, p_{i,2}, \ldots, p_{i,P_i} \}$. We assume $P_i > Q_i$ (the $\mathcal{Z}$-transform is strictly proper), all poles are simple ($p_{i,j} \neq p_{i,l}$ for $j \neq l$)  and \begin{align}
     1 \geq \rho_i = | p_{i,1} | > | p_{i,2} | \geq \cdots \geq | p_{i,P_i} |. \label{eqn:assum_poles} \end{align}
     So the $\mathcal{Z}$-transform of each agent is stable and the maximum absolute value of its poles, denoted $\rho_i$, is attained by a unique pole (the pole denoted by $p_{i,1}$). In particular, the region of convergence of~\eqref{eqn:equations_freq_domain_signal_assumptions} is $\mathcal{ROC}_i = \left\{ z \in \mathbf{C} \colon | z | > \rho_i \right\}$.

\end{assumption}

\section{Extending the running consensus detector (RCD) to the ARMA setup}
\label{sec:running_consensus}
We propose that agents solve cooperatively the global decision problem by resorting to an extension of the running consensus detector (RCD)~\cite{Running_consensus_w_random_networks_2011,running_consensus_original_2009,Enforcing_Consensus}. Our extension retains the generic template of the standard RCD and works as follows: at each time~$k$,  each agent $i$ forms its updated local decision variable $l_{i}(k)\in \mathbf{R}$  by combining in an affine manner the latest decision variables $l_{j}(k-1)$ communicated by its neighbours $j\in \mathcal{N}_i$ with the history of its own local observations 
$$\pmb{y}_i(k) = 
\begin{bmatrix} y_i(0) \\ y_i(1) \\ \vdots \\ y_i(k) \end{bmatrix} \in \mathbf{R}^{k+1}.
$$
Specifically, stack the local decision variables $l_{i}(k)$, $1 \leq i \leq N$, in a global decision vector $l{(k)} \in \mathbf{R}^N$; we suggest that this vector follow the dynamics
\begin{align}
\label{eqn:running_consensus}
l{(k)}&=W\,l{(k-1)}+\eta{(k)},
\end{align} 
 for $k \geq 0$ (with initial conditions $l(-1)=0$), with $\eta(k) \in \mathbf{R}^N$ being an innovation vector where each component~$i$ depends in an affine manner on the history of local observations of agent~$i$, that is,
\begin{align}
\eta_i(k)=\mathcal{A}_{i,k}\big (  \pmb{y}_i(k) \big).
\label{eqn:affine_mapping}
\end{align} 
Here, $\mathcal{A}_{i,k}:\mathbf{R}^{k+1} \mapsto \mathbf{R}$ is an affine map  to be detailed in the next section~\ref{sec:Asymptotic_Detection_Errors}.  
The dynamics~\eqref{eqn:running_consensus} has a simple interpretation: the first term on the right-hand side blends  local decision variables among neighbors, while the second term blends past observations  of each agent with  own latest ones.

Once each agent~$i$ has formed its updated local decision variable $l_i(k)$, it can immediately use it to produce a local decision $D_i(k) \in \{0, 1\}$ about the global environment: 
	\begin{align}
D_{i}(k)=\begin{cases} 1, & \mbox{if } l_{i}(k)\geq \gamma \\
0, & \mbox{if } l_{i}(k)<\gamma, \end{cases}
\label{eqn:LLR_test_distrbuted} 
\end{align}
 where the threshold $\gamma\in \mathbf{R}$ is a design parameter; intuitively, lowering $\gamma$ biases the decision towards $H_1$ and vice-versa. 
 We discuss the choice of $\gamma$  after establishing the long run performance of detectors~\eqref{eqn:LLR_test_distrbuted} (see the discussion immediately after Theorem~\ref{theorem:energy_theorem}). For now, $\gamma$ is just a fixed parameter.
\section{Implementing efficiently the extended RCD}
\label{sec:Asymptotic_Detection_Errors}
In this section, we first propose an affine map $\mathcal{A}_{i,k}$
to be used in~\eqref{eqn:affine_mapping}; we draw  inspiration from the special case of a network with $N = 1$ agent. Because $\mathcal{A}_{i,k}$ acts on an ever-expanding history vector $\pmb{y}_i(k) \in \mathbf{R}^{k+1}$, a naive implementation would demand an equally ever-expanding memory at each agent and would soon become unrealizable. In this section, we bypass this issue by indicating a feasible implementation that runs with finite memory independent of  time $k$.


 \subsection{Motivating $\mathcal{A}_{i,k}$ through a single-agent network }
 Consider a network with $N = 1$ agent. At time~$k$,  the sole agent~$i = 1$  has the history of observations $\pmb{y}_i(k)$ and must decide between 
 \begin{align}
H_1&:\enspace \pmb{y}_i(k) = \pmb{\theta}_i(k) + \pmb{n}_i(k)\\
H_0&:\enspace\pmb{y}_i(k) = \pmb{n}_i(k),
\label{eqn:dec_single}
\end{align}
 where
 \begin{equation}
     \label{eqn:mdefs}
\pmb{\theta}_i(k) = 
\begin{bmatrix} \theta_i(0) \\ \theta_i(1) \\ \vdots \\ \theta_i(k) \end{bmatrix},  \quad \quad
\pmb{n}_i(k) = 
\begin{bmatrix} n_i(0) \\ n_i(1) \\ \vdots \\ n_i(k) \end{bmatrix}.
\end{equation}
Note that the noise vector $\pmb{n}_i(k)$ is Gaussian-distributed, that is, $\pmb{n}_i(k) \sim \mathcal{N}( 0, \pmb{\Sigma}_i(k) )$, where $\pmb{\Sigma}_i(k) \in \mathbf{R}^{(k+1)\times (k+1)}$ denotes its covariance matrix. Standard results from statistics~\cite{stats_book},~\cite{Kay_Book} show that the optimal detector under a Bayes risk criterion is a threshold on the  log-likelihood function, that is, 
 	\begin{align}
 \label{eqn:NP_detector_N_1}
 D_i(k)=\begin{cases} 1, &\mbox{if } l_i(k) \geq \gamma \\
 0, & \mbox{if } l_i(k) <\gamma, \end{cases}
 \end{align} 
 where 
 \begin{align}
 l_i(k) = \pmb{\theta}_i(k)^T \pmb{\Sigma}_i(k)^{-1} \pmb{y}_i(k) - \pmb{\theta}_i(k)^T \pmb{\Sigma}_i(k)^{-1}\pmb{\theta}_i(k)  / 2,\\
 \label{eqn:l1k}
 \end{align}
 assuming, for the moment, invertibility of matrix $\pmb{\Sigma}_i(k)$ (theorem~\ref{theorem:computation_covariance} will show that such assumption is indeed true).
 A key observation now is that
the sequence $\left( l_i(k) \right)_{k \geq 0}$ of these decision variables can be obtained recursively as follows: let  
 \begin{align}
 \pmb{\Sigma}_i(k)=\pmb{L}_i{(k)}\, \pmb{L}_i{(k)}^T
 \label{eqn:Cholesky}
 \end{align}
 be the Cholesky decomposition~\cite{Horn}, with
 $\pmb{L}_i{(k)}\in \mathbf{R}^{(k+1)\times (k+1)}$ being a lower-triangular matrix with positive diagonal entries; using this decomposition in~\eqref{eqn:l1k} gives
 \begin{equation} l_i(k)  =  \widehat{\pmb{\theta}}_i(k)^T  \widehat{\pmb{y}}_i(k) - \widehat{\pmb{\theta}}_i(k)^T \widehat{\pmb{\theta}}_i(k)  / 2, \label{eqn:wl1k} \end{equation}
 where
 \begin{equation}
     \label{eqn:deftheta}
\widehat{\pmb{\theta}}_i(k) = \pmb{L}_i{(k)}^{-1} \pmb{\theta}_i(k)
\end{equation}
 and 
 \begin{equation}
     \label{eqn:defy}
 \widehat{\pmb{y}}_i(k) = \pmb{L}_i{(k)}^{-1} \pmb{y}_i(k).
 \end{equation}
 Because $\pmb{\Sigma}_i(k-1)$ is nested in $\pmb{\Sigma}_i(k)$ as 
 $$
\pmb{\Sigma}_i{(k)} = \begin{bmatrix} \pmb{\Sigma}_i{(k-1)} & \star \\ \star & \star \end{bmatrix},
$$
the Cholesky factor $\pmb{L}_i{(k-1)}$ is also nested in $\pmb{L}_i{(k)}$ as $$
\pmb{L}_i{(k)} = \begin{bmatrix} \pmb{L}_i{(k-1)} & 0 \\ \star & \star \end{bmatrix}.
$$
Therefore,
  \begin{equation} \widehat{\pmb{\theta}}_i(k)  =  
\begin{bmatrix} \pmb{L}_i{(k-1)}^{-1} \pmb{\theta}_i(k-1) \\ \widehat\theta_i(k) \end{bmatrix}
\label{eqn:wt}
\end{equation} and
\begin{equation} \widehat{\pmb{y}}_i(k)  = 
\begin{bmatrix} \pmb{L}_i{(k-1)}^{-1} \pmb{y}_i(k-1) \\ \widehat y_i(k) \end{bmatrix};
\label{eqn:wy} \end{equation}
 plugging now~\eqref{eqn:wt} and~\eqref{eqn:wy} in~\eqref{eqn:wl1k} yields the recursive form
\begin{equation}
l_i(k) = 
l_i(k-1) + \widehat\theta_i(k)\, \widehat y_i(k) - \widehat\theta_i(k)^2 / 2.
\label{eqn:recl1k}
\end{equation}

Comparing updates~\eqref{eqn:recl1k} and~\eqref{eqn:running_consensus} suggests defining $\mathcal{A}_{i,k}$ as
 \begin{align}
 \label{eqn:mapping_Definition}
\mathcal{A}_{i,k}\big(  \pmb{y}_i{(k)}  \big) = \widehat{
	\theta}_i(k)\,\, \widehat{
	y}_i(k) -\frac{1}{2}\widehat{\theta}_i(k)^2,
 \end{align}
 for now a general $N\geq 1$ and agents $i=1\,\dots,N$. This definition insures that, for $N=1$, RDC exactly recovers the log-likelihood scheme since the weight matrix $W$ in~\eqref{eqn:running_consensus} reduces to $W=1$ when $N=1$. In~\eqref{eqn:mapping_Definition} the scalars $\widehat{\theta}_i(k)$ and 
$\widehat{y}_i(k)$ stand for the last components of the vectors $\widehat{\pmb{\theta}}_i(k)$ and $\widehat{\pmb{y}}_i(k)$ (each of size $k+1$) defined in~\eqref{eqn:deftheta} and~\eqref{eqn:defy}.

In conclusion, defining 
\begin{align} 
\widehat y(k) = 
\begin{bmatrix} \widehat y_1(k) \\ \vdots \\ \widehat y_N(k)  \end{bmatrix}, \quad \widehat\theta(k) = 
\begin{bmatrix} \widehat\theta_1(k) \\ \vdots \\ \widehat\theta_N(k)  \end{bmatrix}, 
\label{eqn:tN} \end{align}
and the $N\times N$ diagonal matrix $\widehat \Theta(k) = \text{diag}( \widehat \theta(k) )$, we have that update~\eqref{eqn:running_consensus} consists in 
\begin{align}
l(k) = W l(k-1) + \widehat \Theta(k)\, \widehat y(k) - \frac{1}{2} \widehat \Theta(k)^2 \pmb{1}.
    \label{eqn:newupdate}
\end{align}

 The remainder of this section presents an efficient way to compute  $\widehat{\theta}_i(k)$ and $\widehat{y}_i(k)$ in a recursive manner. The first step is to obtain the Cholesky factors $ \pmb{L}_i{(k)}$ of~\eqref{eqn:Cholesky}  recursively, and proving that these matrices are invertible since they have positive diagonals. Invertibility of $\pmb{\Sigma}_i{(k)}$ then follows by~\eqref{eqn:Cholesky}.  
 
 \subsection{Obtaining recursively the Cholesky decomposition of $\pmb{\Sigma}_i{(k)}$}

Exploiting signal structures, such as ARMA, to devise recursive algorithms has been a key idea in statistical signal processing, e.g., in filtering based on the innovations process~\cite{kailath2000linear}. For example, in Section 4.4 of~\cite{kailath2000linear} Kailath et.al provide efficient implementations of the L-D-L and Cholesky decompositions of an exponentially correlated process.  Similar approaches can be deployed, here, to recursively obtain the Cholesky factor $\pmb{L}_i(k)$ in~\eqref{eqn:Cholesky}, as the next theorem shows.

 \begin{Theorem}(Recursion for the Cholesky factor)
 	\label{theorem:computation_covariance}
  The Cholesky factor $\pmb{L}_i(k)$ in~\eqref{eqn:Cholesky} is given by
 	 $$\pmb{L}_i(k)=	\sigma\,\big(I_{k+1}-\pmb{A}_{i}{(k)}\big)^{-1}\,\pmb{B}_{i}{(k)},$$
 	where the strictly lower-triangular matrices $\pmb{A}_{i}{(k)}$ and the lower-triangular matrices $\pmb{B}_{i}{(k)}$ obey the recursion
 	\begin{align}
 	\pmb{A}_{i}{(k+1)}&=\begin{pmatrix}
 	\pmb{A}_{i}{(k) }  & 0 \\
 	\pmb{\alpha}_i(k+1)^T &  0    \\
 	\end{pmatrix}, \\ \pmb{B}_{i}(k+1)&=\begin{pmatrix}
 	\pmb{B}_{i}(k)   & 0 \\
 	\pmb{\beta}_i(k+1)^T &  1  
 	\end{pmatrix}.
 	\label{eqn:definition_ARMA_covariance}
 	\end{align}
 	The vectors $\pmb{\alpha}_i(k)$ and $\pmb{\beta}_i(k)$ in~\eqref{eqn:definition_ARMA_covariance} start at $\pmb{\alpha}_i(0)=0$ and $\pmb{\beta}_i(0)=1$ and are given by
 	\begin{align} 
 	\pmb{\alpha}_i(k+1)^T&=\left\{
 	\begin{array}{ll}
 	\Big(a_i(k+1) \, \dots\,  a_i(1)\Big)  &, k+1 \leq p_i \\[5pt]
 	\Big(0_{k+1-p_i}^T\,\, a_i(p_i) \, \dots \, a_i(1)\Big) &,  k+1 > p_i,
 	\end{array}
 	\right. \\[5pt]
 	\pmb{\beta}_i(k+1)^T&=\left\{
 	\begin{array}{ll}
 	\Big(b_i(k+1) \, \dots\,  b_i(1)\Big)  &, k+1 \leq q_i \\[5pt]
 	\Big(0_{k+1-q_i}^T\,\, b_i(q_i) \, \dots \, b_i(1)\Big) &,  k+1 > q_i.
 	\end{array}
 	\right.
 	\label{eqn:definition_ARMA_covariance_updates}
 	\end{align}
 \end{Theorem}
 \begin{proof}
 	From~\eqref{eqn:noise_ARMA}, vector $\pmb{n}_i(k)$ in~\eqref{eqn:mdefs} satisfies
 	\begin{align}
 ( I_{k+1} - \pmb{A}_i(k) ) \pmb{n}_i(k) = \sigma\,\pmb{B}_i(k)\, \pmb{\epsilon}_i(k),
 	\end{align}
 	with  matrices $\pmb{A}_i(k)$ and $\pmb{B}_i(k)$ defined as above, and 
  $\pmb{\epsilon}_i(k)$ denoting a zero-mean white Gaussian vector: $\pmb{\epsilon}_i(k) \sim \mathcal{N}( 0, I_{k+1} )$. Because $I_{k+1}-\pmb{A}_i(k)$ is a lower triangular matrix with ones on the main diagonal, its inverse  exists (and is also lower triangular with ones in the main diagonal). This yields
\begin{align}
  \pmb{n}_i(k) = \sigma\,( I_{k+1} - \pmb{A}_i(k) )^{-1} \pmb{B}_i(k)\, \pmb{\epsilon}_i(k).
 	\end{align}
  Thus, the covariance of $\pmb{n}_i(k)$ is
$\pmb{\Sigma}_i(k) = \pmb{L}_i{(k)}\, \pmb{L}_i{(k)}^T$, with $\pmb{L}_i(k)=	\sigma\,\big(I_{k+1}-\pmb{A}_{i}{(k)}\big)^{-1}\,\pmb{B}_{i}{(k)}$, as claimed. Finally, note that $\pmb{L}_i{(k)}$, being the product of two lower-triangular matrices with positive diagonal entries, is itself such a matrix; thus the mentioned $\pmb{L}_i{(k)}$ is indeed the Cholesky factor of $\pmb{\Sigma}_i(k)$.
 \end{proof}
\subsection{Implementing efficiently the map $\mathcal{A}_{i,k}$ }

In implementing the map $\mathcal{A}_{i,k}$  in~\eqref{eqn:mapping_Definition}, the main computational steps are~\eqref{eqn:deftheta} and~\eqref{eqn:defy}, both of which consist in the generic operation
\begin{equation}
     \label{eqn:gop}
\pmb{z}_i(k) = 
\begin{bmatrix} z_i(0) \\ z_i(1) \\ \vdots \\ z_i(k) \end{bmatrix} \quad \mapsto \quad \widehat{\pmb{z}}_i(k) = 
\begin{bmatrix} \widehat z_i(0) \\ \widehat z_i(1) \\ \vdots \\ \widehat z_i(k) \end{bmatrix},
\end{equation}
where $\widehat{\pmb{z}}_i(k) = \pmb{L}_i{(k)}^{-1} \pmb{z}_i(k)$. Were we to store and invert explicitly the ever-growing matrix $\pmb{L}_i{(k)} \in \mathbf{R}^{(k+1)\times (k+1)}$  at each time step $k$, we would soon encounter memory and computational limits. The next theorem shows that $\widehat{\pmb{z}}_i(k)$ can be efficiently obtained from $\pmb{z}_i(k)$ through ARMA filtering.

\begin{Theorem}(ARMA filter interpretation)
\label{theorem:convolution_structure}
The entries of the vector $ \widehat{\pmb{z}}_i(k)$, defined as $\pmb{L}_i{(k)}^{-1} \pmb{z}_i(k)$, are the first $k+1$ output samples of an ARMA$(p_i,q_i)$ filter
which is fed $\pmb{z}_i(k)$, i.e.,
\begin{align}
\label{eqn:theorem_convolution_conclusion}
\widehat{z}_i(k)&=-\sum_{j=1}^{q_i} {b_{i}(j)}{}\,\widehat{z}_i(k-j) + \frac{1}{\sigma }z_i(k)-\sum_{j=1}^{p_i} \frac{a_{i}(j)}{\sigma}\, z_i(k-j),
\end{align}
with zero initial conditions $\widehat{z}_i(k) = z_i(k)=0$ for $k < 0$.
\end{Theorem}
\begin{proof}
From 
\begin{align}
  \underbrace{\sigma\,( I_{k+1} - \pmb{A}_i(k) )^{-1} \pmb{B}_i(k)}_{\pmb{L}_i{(k)}} \widehat{\pmb{z}}_i(k) = \pmb{z}_i(k)
 	\end{align}
follows
\begin{align}
\label{eqn:matrix_equation_convolution_proof}
\pmb{B}_{i}{(k)}\,\widehat{\pmb{z}}_i(k) &= \frac{1}{\sigma} \big(I_{k+1}-\pmb{A}_{i}{(k)}\big)\,\pmb{z}_i{(k)}.
\end{align}
Plugging~\eqref{eqn:definition_ARMA_covariance} and~\eqref{eqn:definition_ARMA_covariance_updates} into~\eqref{eqn:matrix_equation_convolution_proof} gives~\eqref{eqn:theorem_convolution_conclusion}.\end{proof}
The filter in~\eqref{eqn:theorem_convolution_conclusion} is nothing else than the inverse of the noise filter~\eqref{eqn:noise_ARMA}; so, when applied to $\pmb{z}_i{(k)}=\pmb{n}_i(k)$, it outputs  the zero-mean white Gaussian signal $\widehat{\pmb{z}}_i{(k)}=\pmb{\epsilon}_i(k)$. So filter~\eqref{eqn:theorem_convolution_conclusion} essentially  whitens ${\pmb{n}}_i(k)$. 

To conclude, agent~$i$ implements the map $\mathcal{A}_{i,k}$ at each time step $k$, by obtaining the needed $\widehat{\theta}_i(k)$ and $\widehat{y}_i(k)$ as the outputs of its two local copies of the filter~\eqref{eqn:theorem_convolution_conclusion}: one copy is fed the signal of interest $ \theta_i (k) $; the other is fed the stream of local observations $ {y}_i(k) $.

To compute the filtered signals $\widehat{y}_i(k),\widehat{\theta}_i(k)$, agent $i$ stores the past $q_i$ samples of the filtered signals $\widehat{\pmb{{y}}}_i(k)$, $ \widehat{\pmb{{\theta}}}_i(k)$ and the past $p_i$ samples of the unfiltered signals $\pmb{y}_i(k)$, $ \pmb{\theta}_i(k)$, that is,
\begin{align}
&\{\widehat{y}_i(k),\dots,\widehat{y}_i(k-q_i)\},\, \{ \widehat{\theta}_i(k),\dots, \widehat{\theta}_i(k-q_i)\} \\
&\{{y}_i(k),\dots,{y}_i(k-p_i)\},\, \{ {\theta}_i(k),\dots, {\theta}_i(k-p_i)\},
\end{align}
with zero initial conditions for all signals, i.e., $\widehat{y}_i(k)=y_i(k)= \widehat{\theta}_i(k)= {\theta}_i(k)=0$ for any $k<0$. So the memory requirements of agent $i$ grow as $\mathcal{O}({p_i+q_i}+|\mathcal{N}(i)|)$ with $|\mathcal{N}(i)|$ the size of its neighbouring set in communication network $G$ (agent $i$ also needs to store the non-zero entries of the $i$-th row of matrix $W$).   By using~\eqref{eqn:theorem_convolution_conclusion} to compute $ {\theta}_i(k)$ and $\widehat{y}_i(k)$ one easily concludes that implementing~\eqref{eqn:LLR_test_distrbuted} also requires $\mathcal{O}({p_i+q_i}+|\mathcal{N}(i)|)$ operations (sum or products).
\section{Theoretical analysis of RCD}
\label{sec:theory_on_RCD}
In this section we analyze the performance of the local RCD detectors~\eqref{eqn:LLR_test_distrbuted}.
Specifically, for the local detector of agent~$i$, we study the long time behaviour ($k \rightarrow \infty$) of the two error probabilities, that is, the probabilities of false alarm and miss,
\begin{align}
\mathbb{P}_{F,i}(k)&:=\mathbb{P}_{H_0}( D_i(k) = 1),
\label{eqn:type_I_error}
\\
\mathbb{P}_{M,i}(k)&:=\mathbb{P}_{H_1}( D_i(k) = 0).
\label{eqn:type_II_error}
\end{align}
The asymptotic behaviour of these probabilities will be expressed in terms of a time-scaling sequence $( f(k) )_{k \geq 0}$, which depends on which regime (a) or (b) is activated: in regime (a), $f(k) $ is a linearly growing sequence with $f_a(k) = k+1$, while in regime (b) it is the constant sequence with $f_b(k) = 1$. Hereafter, we denote by ${\mathcal F} = \left\{ ( f_a(k) )_k, ( f_b(k) )_k \right\}$ the two previously mentioned time scalings.

Which regime is activated depends on  the growing rate of \begin{equation} \left( \left\| \widehat{\pmb{\theta}}(k) \right\|^2 \right), \label{eqn:energy} \end{equation} where
\begin{align}
    \widehat{\pmb{\theta}}(k) = \begin{bmatrix} \widehat{\pmb{\theta}}_1(k) \\ \vdots \\ \widehat{\pmb{\theta}}_N(k) \end{bmatrix} \in \mathbf{R}^{N(k+1)}, \label{eqn:concat} \end{align}
with $\widehat{\pmb{\theta}}_i(k)$ defined in~\eqref{eqn:deftheta}. Note that $$\left\| \widehat{\pmb{\theta}}(k) \right\|^2 = \sum_{m = 0}^k \left\| \widehat \theta(m) \right\|^2,$$
where $\widehat \theta(k)$ is as in~\eqref{eqn:tN};
thus, $\left\| \widehat{\pmb{\theta}}(k) \right\|^2$ can be interpreted as the cumulative energy of the whitened network signal $(  \widehat{\pmb{\theta}}(m) )$ from $m = 0$ till $m = k$.

We break up the analysis of the error probabilities in two steps, summarized as follows:
\begin{enumerate}
    \item Assuming, momentarily, that the sequence~\eqref{eqn:energy} grows as 
    \begin{equation} \left\| \widehat{\pmb{\theta}}(k) \right\|^2 = \alpha f(k) + o\left( f(k) \right), \label{eqn:assumpg} \end{equation}
    where $\alpha$ is a positive number and the sequence $\left( f(k) \right)$ is one of the time-scalings in $\mathcal{F}$, theorem~\ref{theorem:energy_theorem} shows that \begin{align}
-\frac{\log \mathbb{P}_{F,i}(k)}{f(k)} & \rightarrow  \beta_F \label{eqn:asympF} \\
-\frac{\log \mathbb{P}_{M,i}(k)}{f(k)} & \rightarrow  \beta_M, \label{eqn:asympM}
    \end{align}
    as $k \rightarrow \infty$, for some explicit constants $\beta_F,\beta_M>0$. Thus, in regime (a), the probabilities $\mathbb{P}_{F,i}(k)$ and $\mathbb{P}_{M,i}(k)$ decay exponentially fast to zero as in standard large deviations results~\cite{ellis2006entropy}, roughly behaving as $\mathbb{P}_{F,i}(k) \simeq \exp(-\beta_F (k+1))$ and $\mathbb{P}_{M,i}(k) \simeq \exp(-\beta_M (k+1))$. This exponential vanishing of the probabilities can be attributed to the linear increase of the energy of $\widehat{\pmb{\theta}}(k)$, a generic effect already noted in centralized setups~\cite{Poor_Vincent_estimation_book},~\cite{Levy_book},~\cite{van_trees_book_1}. 
    In sharp contrast, for the novel regime (b), the probabilities of error settle into a positive error floor, so $\mathbb{P}_{F,i}(k) \simeq \exp(-\beta_F)$ and $\mathbb{P}_{M,i}(k) \simeq \exp(-\beta_M)$ for large $k$;
    \item In Theorem~\ref{theorem:single_agent_energy}, we show that the assumption~\eqref{eqn:assumpg} indeed holds by giving explicitly both the positive constant $\alpha$ and the time scaling $\left( f(k) \right) \in {\mathcal F}$ that arise from the underlying ARMA filters.
\end{enumerate}

We now carry out the two steps of the analysis.

\subsection{Step 1: asymptotic performance of $\mathbb{P}_{F,i}(k)$ and $\mathbb{P}_{M,i}(k)$}

\label{sec:general_section}

The key result is Theorem~\ref{theorem:energy_theorem}, described next.

\begin{Theorem}(Asymptotic performance of the local detectors~\eqref{eqn:LLR_test_distrbuted})
	\label{theorem:energy_theorem} 
Let Assumptions~1 and~2 in Section~\ref{sec:assumptions} hold. 

If the growing rate condition~\eqref{eqn:assumpg} also holds, then the false alarm probability $\mathbb{P}_{F,i}(k)$ and the miss probability $\mathbb{P}_{M,i}(k)$ of the local detector~\eqref{eqn:LLR_test_distrbuted} at agent~$i$ obey~\eqref{eqn:asympF} and~\eqref{eqn:asympM} with
	\begin{align}	
 \beta_F & = \begin{cases}  
	\dfrac{\alpha}{8} &, \mbox{ if } f = f_a \\
 -\log\, \mathcal{Q} \left( \dfrac{\sqrt{\alpha}}{2} +\dfrac{ \gamma\,N}{\sqrt{\alpha}} \right) &,\mbox{ if }   f = f_b,  \end{cases} \label{eqn:energy_theorem_resulta} \\[5pt] 	\beta_M& = \begin{cases} 
\dfrac{\alpha}{8} &, \mbox{ if } f = f_a \\
-\log\, \mathcal{Q}\left( \dfrac{\sqrt{\alpha}}{2}-\dfrac{ \gamma\,N}{\sqrt{\alpha}} \right) &, \mbox{ if } f = f_b, \end{cases}
		\label{eqn:energy_theorem_resultb}
	\end{align}
 where the $\mathcal{Q}$ function computes the tail probability of a standard Gaussian,  $\mathcal{Q}(t) = \int_{t}^{+\infty} \exp(-\frac{x^2}{2}) dx / \sqrt{2\pi}$. 
\end{Theorem}

Before we turn to the proof of the theorem it is worth of note that, as~\eqref{eqn:energy_theorem_resulta} and~\eqref{eqn:energy_theorem_resultb} show, the asymptotic performance of the error probabilities  $\mathbb{P}_{F,i}(k)$ and  $\mathbb{P}_{M,i}(k)$ of agent~$i$ turn out not to depend on the specific agent~$i$. So all agents implementing RDC will enjoy the same asymptotic performance, regardless of temporal correlation. 

Note also that Assumption~3, which attaches ARMA structures to both the detection signal and measurement noise, is not invoked here. This means that Theorem~\ref{theorem:energy_theorem} is actually stronger than it needed to be, for it  covers a class of temporally correlated signals and noise that extend well beyond the ARMA ones; this result might be of independent interest.


\label{sec:theory_Section}


\begin{proof}
We prove only the asymptotic behaviour of the probability of miss~\eqref{eqn:asympM} (proving result~\eqref{eqn:asympF} is similar). So, from now on, we assume that the measurements $ y(k) $ are being generated according to hypothesis $H_1$, see~\eqref{eqn:hypothesis_formulation}.

As the probability of miss in~\eqref{eqn:type_II_error}
depends on the distribution of the local decision $D_i(k)$, which in turn depends on the distribution of $l_i(k)$ (recall~\eqref{eqn:LLR_test_distrbuted}), we start by scrutinizing the distribution of vector $l(k)$, which evolves according to~\eqref{eqn:newupdate}.

Because $\widehat y(k)$ is a Gaussian random vector, it follows from~\eqref{eqn:newupdate} that $l(k)$ is likewise a Gaussian random vector: \begin{align} l(k) \sim {\mathcal N}(\mu(k), \Omega(k)), \label{eqn:vecl} \end{align}
the symbols $\mu(k)$ and $\Omega(k)$  denoting the mean and covariance of $l(k)$. 
In fact, from~\eqref{eqn:newupdate}, we can derive a recursion for the means and covariances, using the $H_1$ hypothesis:
\begin{align}
\mu(k+1) & =  W \mu(k) +  \frac{1}{2} \widehat\Theta(k)^2 \pmb{1} \label{eqn:r_mean} \\ 
\Omega(k+1)  & =  W \Omega(k) W^T + \widehat \Theta(k)^2,
\label{eqn:r_cov}
\end{align}
where, to obtain~\eqref{eqn:r_mean}, we used
$\mathbb{E} \left( \widehat y(k) \right) =  \widehat\Theta(k) \pmb{1}$ (since $H_1$ is being assumed), while, for~\eqref{eqn:r_cov}, 
 we used the fact that $\left( \widehat y(k) \right)$ is a (whitened) uncorrelated sequence of random vectors with the identity as their covariance matrix (recall~\eqref{eqn:Cholesky} and~\eqref{eqn:defy}).

Our next step consists in exploiting 
the recursions~\eqref{eqn:r_mean} and~\eqref{eqn:r_cov} to 
establish the asymptotic behaviour of $\mu(k)$ and $\Omega(k)$. For this, the next lemma is instrumental.
\begin{lemma}
	\label{lemma:aux1}
	Suppose the sequence of $n$-dimensional vectors $( v(k) )_{k \geq 0}$ obeys the recursion  \begin{align} v(k+1) = \mathcal{W} v(k) + u(k), \label{eqn:template} \end{align} for $k \geq 0$ (with $v(-1) = 0$), where $\mathcal{W}$ is a symmetric matrix with eigenvalue decomposition
 \begin{align} \label{eqn:gevd}
	 \mathcal{W} =\begin{bmatrix}
	\tfrac{1}{\sqrt{n}}\mathbf{1}_n &  \mathcal{S}
	\end{bmatrix}\begin{bmatrix}
	1 &  \\
	& \mathcal{D}
	\end{bmatrix} \begin{bmatrix}
	\tfrac{1}{\sqrt{n}}\mathbf{1}_n^T\\
	\mathcal{S}^T
	\end{bmatrix},
	\end{align}
	with ${\mathcal D} \in \mathbf{R}^{(n-1)\times (n-1)}$ being a diagonal matrix whose diagonal entries satisfy
	$| d_j | < 1$, for $1 \leq j \leq n-1$.

 If $(u(k))$ is a sequence of non-negative vectors ($u_j(k) \geq 0$ for $1 \leq j \leq n$) that grows as
 \begin{align} \sum_{m = 0}^k \pmb{1}_n^T u(m) = \theta f(k) + o( f(k) ), \label{eqn:ggr} \end{align}
 where $\theta > 0$ and $( f(k) )$ is one of the time-scalings in $\mathcal{F}$, then
 \begin{align} \frac{v(k)}{f(k)} \rightarrow \dfrac{\theta}{n} \pmb{1}_n, \label{eqn:lemmaconc} \end{align}
 as $k \rightarrow \infty$.
\end{lemma}

Lemma~\ref{lemma:aux1} is proved in the Appendix. We now use it twice:
\begin{itemize}
\item First, note that~\eqref{eqn:r_mean} fits the template~\eqref{eqn:template} when we identify $n = N$, $v(k) = \mu(k)$, $\mathcal{W} = W$, and $u(k) = \widehat \Theta(k)^2 \pmb{1} / 2$; indeed, the eigenvalue decomposition~\eqref{eqn:mW} verifies~\eqref{eqn:gevd}, and~\eqref{eqn:assumpg} verifies~\eqref{eqn:ggr} with $\theta = \alpha/2$. Thus,
\begin{align}
	\frac{\mu(k)}{{f(k)}} &\rightarrow \frac{\alpha}{2 N}\,\mathbf{1}_N.
	\label{eqn:limit_mean_proof}
 \end{align}
 \item Second, note that~\eqref{eqn:r_cov}, too,  fits the template~\eqref{eqn:template} once we switch to the vectorized form of~\eqref{eqn:r_cov}:
 \begin{align}
\text{vec}\left( \Omega(k+1) \right) = \left( W \otimes W \right) \text{vec}\left( \Omega(k) \right) + \text{vec}\left( \widehat{\Theta}^2(k) \right), \label{eqn:vecf}
 \end{align}
where $\otimes$ denotes Kronecker product and the generic property $\text{vec}( A X B ) = ( B^T \otimes A ) \text{vec}( X )$ was used to go from~\eqref{eqn:r_cov} to~\eqref{eqn:vecf}. Recursion~\eqref{eqn:vecf} fits the template~\eqref{eqn:template} with the obvious identifications. Condition~\eqref{eqn:gevd} is verified because~\eqref{eqn:mW} implies that $\mathcal{W} = W \otimes W$  has the eigenvalue decomposition $$\mathcal{W} = \underbrace{( Q_W \otimes Q_W)}_{Q_\mathcal{W}}  \underbrace{( \Lambda_W \otimes \Lambda_W)}_{\Lambda_\mathcal{W}} \underbrace{( Q_W \otimes Q_W)^T}_{Q_\mathcal{W}^T},$$ a decomposition that satisfies the eigenvalue and eigenvectors requirements in~\eqref{eqn:gevd}. Furthermore, condition~\eqref{eqn:ggr} is verified with $\theta = \alpha$, thanks to~\eqref{eqn:assumpg}.
Thus,
\begin{align}
 \frac{\text{vec}\left( \Omega(k) \right)}{f(k)} & \rightarrow \frac{\alpha}{N^2}  \pmb{1}_{N^2}.
 \label{eqn:limit_variance_proof}
	\end{align}
\end{itemize}

We are now in position to derive~\eqref{eqn:asympM}, that is, the asymptotic behaviour of \begin{align}-\frac{\log \mathbb{P}_{M,i}(k)}{f(k)}. \label{eqn:ratio} \end{align}
Given that $\mathbb{P}_{M,i}(k) = \mathbb{P}_{H_1}( l_i(k) < \gamma )$ and $l_i(k) \sim {\mathcal N}\left( \mu_i(k), \Omega_{ii}(k) \right)$, where $\mu_i(k)$ is the $i$th component of the vector $\mu(k)$ and $\Omega_{ii}(k)$ the $i$th diagonal entry of the matrix $\Omega(k)$, we have
\begin{align} -\frac{\log \mathbb{P}_{M,i}(k)}{f(k)} 
= - \frac{\log\, \mathcal{Q}\left( \frac{-\gamma + \mu_{i}(k)}{\sqrt{\Omega_{ii}(k)}} \right)}{f(k)}.\label{eqn:ratio2} \end{align}
We study the limit of~\eqref{eqn:ratio2} as $k \rightarrow \infty$ for the two possible regimes:
\begin{itemize}
    \item In regime (a), we have $f(k) = k+1$, meaning that~\eqref{eqn:ratio2} becomes
    \begin{align}
-\frac{\log \mathbb{P}_{M,i}(k)}{f(k)} 
=       - \frac{\log\, \mathcal{Q}\left( \frac{-\gamma + \mu_{i}(k)}{\sqrt{\Omega_{ii}(k)}} \right)}{k+1}. \label{eqn:ad}
    \end{align}
    Writing $$\frac{-\gamma + \mu_{i}(k)}{\sqrt{\Omega_{ii}(k)}} = \underbrace{\frac{-\gamma/(k+1) + \mu_{i}(k)/(k+1)}{\sqrt{\Omega_{ii}(k)/(k+1)}}}_{u(k)} \underbrace{\sqrt{k+1}}_{\sqrt{v(k)}},$$
    we have $u(k) \rightarrow \sqrt{\alpha}/2$ (using~\eqref{eqn:limit_mean_proof} and~\eqref{eqn:limit_variance_proof}) and $v(k) \rightarrow \infty$. Thus, we can invoke the following auxiliary Lemma~\ref{lemma:aux2}, whose proof is in the Appendix.
    \begin{lemma}
        \label{lemma:aux2} Suppose $( u(k) )$ and $(v (k) )$ are sequences with positive terms. If $u(k) \rightarrow \theta > 0$ and $v(k) \rightarrow \infty$, then
    \begin{align} - \frac{\log \mathcal{Q}\left( u(k) \sqrt{ v(k)} \right)}{v(k)}  \rightarrow \frac{\theta^2}{2}. \label{eqn:convQ} \end{align}
    \end{lemma}
    Applying Lemma~\ref{lemma:aux2} in~\eqref{eqn:ad} yields~\eqref{eqn:asympM} with $\beta_M = \alpha / 8$;

\item In regime (b), we have $f(k) = 1$. Therefore, from~\eqref{eqn:limit_mean_proof} and \eqref{eqn:limit_variance_proof}, together with the continuity of the function $\mathcal{Q}$,  the limit~\eqref{eqn:asympM} follows with $\beta_M = -\log \mathcal{Q}\left( \sqrt{\alpha}/2 - \gamma N / \sqrt{\alpha} \right).$ 
\end{itemize}

\end{proof}
\textbf{On threshold $\gamma$.} As~\eqref{eqn:energy_theorem_resulta} and~\eqref{eqn:energy_theorem_resultb} show,
 the threshold $\gamma$ of the local detectors~\eqref{eqn:LLR_test_distrbuted} does not affect the asymptotic behaviour of the probabilities of error for regime (a). This is consistent with previous results in distributed detection~\cite{Running_consensus_w_random_networks_2011} where $\gamma$ is usually set to zero. In contrast, for regime (b) the threshold $\gamma$ has a noticeable impact. It becomes a design parameter that can be exploited to navigate the trade-off between the asymptotic probability of false alarm $\mathbb{P}_{F}(\gamma) := \lim\limits_k \mathbb{P}_{F,i}(k)$ and the asymptotic probability of miss $\mathbb{P}_{M}(\gamma):=\lim\limits_k\mathbb{P}_{M,i}(k)$.
 Indeed, a higher value of $\gamma \geq 0$  bias the detector towards choosing hypothesis $H_0$, thus decreasing the probability of false alarm and increasing the probability of miss: 
	\begin{align}
	\lim_{\gamma\rightarrow +\infty} \mathbb{P}_{F,i}(\gamma)= 0,\enspace\lim_{\gamma\rightarrow +\infty} \mathbb{P}_{M,i}(\gamma)= 1.
	\end{align}
  A symmetric argument holds when $\gamma < 0$.

  ~\\[5pt]
\textbf{A special case of Theorem~\ref{theorem:energy_theorem}.} In the special case of $N = 1$ agent detecting a constant signal $\theta(k) = A$ in white Gaussian noise $n(k) \sim \mathcal{N}\left( 0 , \sigma^2 \right)$, regime (a) is activated because it can be checked that~\eqref{eqn:assumpg} holds with $\alpha = A^2 / \sigma^2$ and $f(k) = k+1$.
  Theorem~\ref{theorem:energy_theorem} therefore asserts  the exponential decay of the probabilities of error with the common exponent $\beta_F = \beta_M = \beta$, where $\beta = A^2/(8\sigma^2)$. Such $\beta$ is in fact the
   Chernoff distance between the distributions of the mesurements $y(k)$ under $H_1$ and $H_0$~\cite{Running_consensus_w_random_networks_2011},~\cite{Large_Deviation_book}. 

\subsection{Step 2: growth of energy $||\widehat{\pmb{\theta}}(k)||^2$ for ARMA models}
\label{sec:perfomance_arma_models}
In this section we derive the growing rate of the energy of $\widehat{\pmb{\theta}}(k)$, defined in~\eqref{eqn:concat}. Specifically, we show that \eqref{eqn:assumpg} holds by getting both $\alpha$ and the time-scaling $( f(k) )$ from the underlying ARMA models.

We start with the observation that~\eqref{eqn:concat} implies that the energy decouples across agents, that is,
\begin{align} \left\| \widehat{\pmb{\theta}}(k) \right\|^2 = \sum_{i = 1}^N \left\| \widehat{\pmb{\theta}}_i(k) \right\|^2, \label{eqn:edecouple} \end{align}
which suggests looking at the growing rate of the energy of the whitened signal at each agent. Thus, hereafter, we focus on a generic agent~$i$ and study the growing rate of $$\left\| \widehat{\pmb{\theta}}_i(k) \right\|^2 = \sum_{m = 0}^k \widehat\theta_i(m)^2.$$

Recall that $ \widehat\theta_i(k) $ is the output of the ARMA filter~\eqref{eqn:theorem_convolution_conclusion} with input $ \theta_i(k) $ and $ \theta_i(k) $ is the impulse response of the ARMA filter~\eqref{eqn:signal_ARMA}. Thus, the $\mathcal{Z}$-transform of $ \widehat \theta_i(k) $ is given by 
\begin{align}
\label{eqn:equations_freq_domain_signal}
 \widehat{\Theta}_i(z^{-1})&=\dfrac{A}{\sigma}\frac{1-\sum\limits_{j=1}^{q_i} a_{i}(j)\,z^{-j}}{1+\sum\limits_{j=1}^{p_i} b_{i}(j)\,z^{-j}}\, \frac{1+\sum\limits_{j=1}^{\bar{q}_i} \bar{b}_{i}(j)\,z^{-j}}{1-\sum\limits_{j=1}^{\bar{p}_i} \bar{a}_{i}(j)\,z^{-j}}.
\nonumber \\
\end{align}
This is just the expression in~\eqref{eqn:equations_freq_domain_signal_assumptions} to which Assumption 3 in Section~\ref{sec:assumptions} applies. Therefore,
from the partial fraction expansion result~\eqref{eqn:partial_fraction_decomposition}, we have
	\begin{align}
	\label{eqn:closed_form_proof}
	 \widehat{\theta}_i(k) = 
	\sum_{j=1}^{P_i} r_{i,j}\,p_{i,j}^k,
	\end{align} 
 where $r_{i,j}$ denotes the residual associated with pole~$i$, that is,
$$r_{i,j } = (1-p_{i,j}\,z^{-1})\, \widehat{\Theta}_i(z^{-1}) \Big|_{z=p_{i,j}.}$$

We  now  evaluate the energy of $\widehat{\pmb{\theta}}_i(k)$, which depends critically on $\rho_i$, the maximum absolute value of the poles (recall equation~\eqref{eqn:assum_poles}).

\begin{Theorem}(Asymptotic energy of  signal $\widehat{\pmb{\theta}}_i(k) $) 
\label{theorem:single_agent_energy} Under Assumption~3 we have
\begin{align}	
\left\| \widehat{\pmb{\theta}}_i(k) \right\|^2 & = \begin{cases}  
	| r_{i,1} |^2 (k+1) + o(k+1) &, \mbox{ if } \rho_i = 1  \\
 \sum\limits_{j, l = 1}^{P_i} \dfrac{r_{i,j} r_{i,l}^\star}{1 - p_{i,j} p_{i,l}^\star} + o( 1 ) &,\mbox{ if }   \rho_i < 1,  \end{cases} \\ \label{eqn:rho1}
		\end{align}
  where $z^\star$ stands for the complex-conjugate of a complex $z$. 
\end{Theorem}
\begin{proof}
  Starting with the case $\rho_i = 1$, we have \begin{align} 1 = | p_{i,1} | > r  \geq | p_{i,2} | \geq \cdots \geq | p_{i,P_i} |, \label{eqn:case1} \end{align} for some $1 > r \geq 0$. Introducing the notation $$\Delta(k) = \sum\limits_{j = 2}^{P_i} r_{i,j} p_{i,j}^k$$ we rewrite  ~\eqref{eqn:closed_form_proof} as
  $\widehat\theta_i(k) = r_{i,1} p_{i,1}^k + \Delta(k)$. This yields
     \begin{align} 
     \left\| \widehat{\pmb{\theta}}_i(k) \right\|^2  = \sum_{m = 0}^k | \widehat\theta_i(m) |^2  &= \sum_{m = 0}^k \left| r_{i,1} p_{i,1}^m + \Delta(m) \right|^2 \\ & = | r_{i,1} |^2 (k+1) + u(k), \label{eqn:decomp}
     \end{align}
     where $$u(k) = \sum\limits_{m = 0}^k r_{i,1} p_{i,1}^m \Delta(m)^\star + r_{i,1}^\star (p_{i,1}^\star)^m \Delta(m) + | \Delta(m) |^2.$$
     
Now, note that $ u(k) $ is of  the form
     $$u(k) = \sum_{l = 1}^L \left( \sum_{m = 0}^k v_l w_l^m \right),$$ (for some $L$ and constants $v_l$) with each $w_l$ being the product (up to conjugates) of two poles not both equal to $p_{i,1}$; in concrete $w_l\in\{p_{i,j_1}\, p_{i,j_2},\,p_{i,j_1}^*\,p_{i,j_2}, \, p_{i,j_1}\,p_{i,j_2}^*,\, p_{i,j_1}^*\,p_{i,j_1}^*\}$ for some $(j_1,j_2)\neq (1,1)$. Thus, from~\eqref{eqn:case1} we have $| w_l | \leq r < 1$, which implies \begin{align}
| u(k) | & \leq  \sum_{l = 1}^L \left( \sum_{m = 0}^k | v_l | r^m \right) \leq \sum_{l = 1}^L | v_l | \frac{1}{1 - r}.
         \end{align}
         The sequence $( u(k) )$ is therefore bounded, and, as a consequence, satisfies $u(k) = o(k+1)$. This observation, together with~\eqref{eqn:decomp}, yields the first case in~\eqref{eqn:rho1}.

  We now turn to the case $\rho_i < 1$. We have \begin{align}
\left\| \widehat{\pmb{\theta}}_i(k) \right\|^2 = \sum_{m = 0}^k | \widehat\theta_i(m) |^2 = \sum_{j, l = 1}^{P_i} \left( \sum_{m = 0}^k r_{i,j} r_{i,l}^\star ( p_{i,j}  p_{i,l}^\star)^m \right).
    \end{align}

 Because $| p_{i,j} p_{i, l}^\star | \leq \rho_i^2 < 1$, it follows that 
 $$\sum_{m = 0}^k r_{i,j} r_{i,l}^\star ( p_{i,j}  p_{i,l}^\star)^m \rightarrow  \dfrac{r_{i,j} r_{i,l}^\star}{1 - p_{i,j} p_{i,l}^\star}.$$
 Therefore, $$\sum_{m = 0}^k | \widehat\theta_i(m) |^2 \rightarrow \sum_{j, l = 1}^{P_i} \dfrac{r_{i,j} r_{i,l}^\star}{1 - p_{i,j} p_{i,l}^\star},$$ which corresponds to the second case in ~\eqref{eqn:rho1}.
\end{proof}

Using~\eqref{eqn:edecouple} and Theorem~\ref{theorem:single_agent_energy}, we conclude that~\eqref{eqn:assumpg} holds. Indeed, let \begin{align}\rho = \max\{ \rho_i \colon 1 \leq i \leq N \}~\label{eqn:defrho} \end{align}
denote the maximum absolute value of the poles~\eqref{eqn:equations_freq_domain_signal} across agents. Two regimes appear:
\begin{itemize}
\item if $\rho = 1$ (thus, $\rho_i = 1$ for at least one agent~$i$), then regime (a) is activated.
In this case, \begin{align} \alpha = \sum_{i \in {\mathcal I}} | r_{i,1} |^2, \label{eqn:inf} \end{align}
where ${\mathcal I} = \{ i \in \{ 1, \ldots, N \} \colon \rho_i = 1 \}$, 
and $f(k) = k+1$;
\item if $\rho < 1$ (thus, $\rho_i < 1$ for all $i$), then regime (b) is activated with 
\begin{align} \alpha = \sum_{i \in \mathcal{I}} \sum_{j, l = 1}^{P_i} \dfrac{r_{i,j} r_{i,l}^\star}{1 - p_{i,j} p_{i,l}^\star} \label{eqn:regBI} \end{align}
where $\mathcal{I} = \{ 1, 2, \ldots, N \}$, and $f(k) = 1$.
\end{itemize}

\subsection{Weakly informative agents}
\label{sec:weak_informative}
As~\eqref{eqn:inf} shows, in regime (a) only the agents in the subset ${\mathcal I}$ contribute to the asymptotic exponential decay rate of the probability of errors enjoyed by all agents. So, the subset ${\mathcal I}$ can be interpreted as the subset of \emph{strongly informative} agents, whereas agents not in this set (if any) may be viewed as \emph{weakly informative} agents---their measurements cannot help improve the exponential decay rate.
This phenomenon, arising from the ARMA setup, means that we may relieve the weakly informative agents from getting measurements, thus saving possibly expensive sensors, while keeping the original asymptotic detection performance.

Note that in regime (b) all agents are strongly informative because they all contribute to the probability error floor, as shown in~\eqref{eqn:regBI}.

\section{Numerical Experiments: DC level embedded in colored noise}
\label{sec:numerical_section}
To validate the theoretical results, 
we consider a detection problem with $N = 40$ agents linked through a communication network (sampled from a Erdós-Renyi model) as shown in Figure~\ref{fig:graph_layout}. 
\begin{figure}[h!]
	\centering
	\includegraphics[width=5cm]{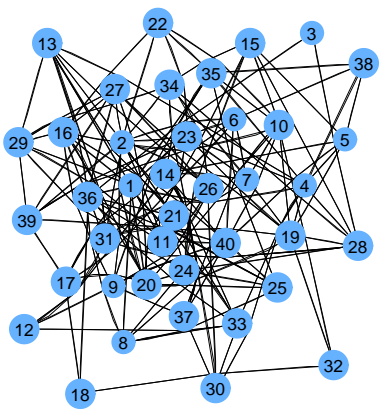}
	\caption{The communication network for section~\ref{sec:numerical_section} is a sample of a Erdós-Renyi graph with $N = 40$ nodes and probability of connection  $2\,\log(N)/N$; this network sample has $154$ edges.}
	\label{fig:graph_layout}
\end{figure}
The matrix $W$ in~\eqref{eqn:running_consensus} is chosen as the distributed averaging matrix~\cite{xiao2004fast}
\begin{align}
W=I_N- \frac{1}{\lambda_2 + \lambda_N} \mathcal{L},
\end{align}
where $\mathcal{L}$ is the Laplacian matrix of the graph in Figure~\ref{fig:graph_layout} and 
$0 = \lambda_1 < \lambda_2 \leq \dots \leq \lambda_N$
denote the  sorted eigenvalues of the Laplacian $\mathcal{L}$. All agents implement detector~\eqref{eqn:LLR_test_distrbuted} with $\gamma = 0$. 

We consider a detection setup in which  agents want to detect a common DC level signal masked in colored noise, that is,
\begin{align}
H_1:\enspace y_i(k)&= A+n_i(k)\\
H_0:\enspace y_i(k)&=n_i(k),
\label{eqn:DC_level}
\end{align}
with $A>0$ the amplitude of the DC level. Although our framework is rich enough to also encompass ARMA signals, we consider a DC level such that the derived rates are simple and interpretable. The noise sequence $( n_i(k) )$ follows an ARMA model 
with
\begin{align}
n_i(k)&=n_i(k-1)+\sigma\, \epsilon_i(k)+\sigma\,b_i(1)\, \epsilon_i(k-1), 
\label{eqn:DC_level_noise}
\end{align}
 which means that~\eqref{eqn:equations_freq_domain_signal} becomes
\begin{align}
 \widehat{\Theta}_i(z^{-1})
&=\frac{A/\sigma}{1+b_i(1)\,z^{-1}}.
\label{eqn:noise_model_Z_domain}
\end{align}
The coefficients $b_i(1)$ of the ARMA filters vary from agent to agent and are chosen so as to trigger regime (a) and regime (b), as detailed in the following sections.


\subsection{Regime (a): error probabilities decay exponentially fast}
\label{sec:reg_a}

We fixed $A = 1$ and $\sigma = 10$ and started by generating the coefficients $b_i(1)$ 
 uniformly at random from the interval $(-1,1)$: this makes the absolute value of the pole of~\eqref{eqn:noise_model_Z_domain}
 strictly less than one for all $i$. Then, to activate regime (a), we modified the coefficient  of agent~$3$ to $b_3(1) = 1$, which makes $\rho_3 = 1$ and, therefore, $\rho = 1$ (recall definition~\eqref{eqn:defrho}).
 
 As a consequence, the subset of strongly informative agents is the singleton $\mathcal{I} = \{ 3 \}$.
 This means that, asymptotically, the energy of the network signal 
$\widehat{\pmb{\theta}}(k)$ in~\eqref{eqn:edecouple}
 comes from the energy of the local signal $\widehat{\pmb{\theta}}_3(k)$, which itself behaves in accord with Theorem~\ref{theorem:single_agent_energy} as $$\left\| \widehat{\pmb{\theta}}_3(k) \right\|^2 = \frac{A^2}{\sigma^2} (k+1) + o (k+1).$$
 (the residue $r_{3,1}$ appearing in~\eqref{eqn:rho1} is equal to $A/\sigma$ for the current case~\eqref{eqn:noise_model_Z_domain}).
Thus, $$\left\| \widehat{\pmb{\theta}}(k) \right\|^2 = \alpha\,(k+1) + o (k+1),$$
 with $\alpha = A^2 / \sigma^2$.

Figure~\ref{fig:infinite_case} shows the empirical probability of miss detection for two agents: the strongly informative agent $3$ and---to illustrate a weakly informative agent---agent $18$. 
\begin{figure}[h!]
	\centering
	\includegraphics[width=7.6cm]{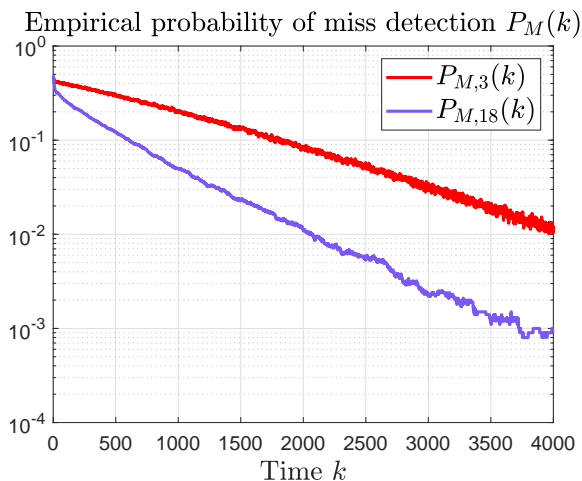}
	\caption{
		 Empirical  probability of miss detection for agent $3$ (strongly informative) and agent~$18$ (weakly informative). All plots are averages over $10^5$ Monte Carlo trials.   The curves $\mathbb{P}_{M,i}(k)$ for the remaining $38$ weakly informative agents are not shown since they are almost identical  to  $\mathbb{P}_{M,18}(k)$.}
	\label{fig:infinite_case}
\end{figure}

In line with the prediction of Theorem~\ref{theorem:energy_theorem},
 Figure~\ref{fig:infinite_case} shows that the probabilities $\mathbb{P}_{M,i}(k)$ decay exponentially fast and satisfy
\begin{align}
\ -\frac{\log \mathbb{P}_{M,i}(k)}{k} \rightarrow \beta_M = \frac{\alpha}{8} = \frac{A^2}{8 \sigma^2} \approx 0.0013.
\label{eqn:risk_convergence}
\end{align}
(At time $k=4000$, the standard deviation of rates $\left\{-\log \mathbb{P}_{M,i}(k)/k \colon 1 \leq i \leq N \right\}$  is approximately $10^{-4}$, and the deviation between the mean rate $-\log \mathbb{P}_{M,i}(k)/k$ and limit $\beta_M$ in~\eqref{eqn:risk_convergence} is approximately $5\times 10^{-4}$). The decay rate in~\eqref{eqn:risk_convergence} is intuitive since it depends on the signal-to-noise ratio $A/\sigma$.

 \subsection{Regime (b): error probabilities settle in a positive floor}
\label{sec_b}

 We fixed $A = 1$ and $\sigma = 5$ and kept all  parameters $b_i(1)$ the same (that is, as in the previous section~\ref{sec:reg_a}) except that we generated $b_3(1)$ uniformly at random in the interval $(-1, 1)$. This means that the poles in~\eqref{eqn:noise_model_Z_domain} have now absolute value strictly less than $1$, that is, $\rho$ in~\eqref{eqn:defrho} satisfies $\rho < 1$; in fact, for our set of $b_i(1)$'s, we have $\rho = 0.94$. Regime (b) is thus activated and all agents are strongly informative.

 As a consequence, Theorem~\ref{eqn:assumpg} predicts that the error probabilities no longer decay to zero; rather, they settle in a positive error floor. For instance, Theorem~\ref{eqn:assumpg} asserts that \begin{align} \mathbb{P}_{F,i}(k) \rightarrow \mathcal{Q}\left(\dfrac{\sqrt{\alpha}}{{2}} \right) \label{eqn:conf} \end{align}
for all agents~$i$; here, $\alpha$ is given by~\eqref{eqn:regBI}, which, in view of~\eqref{eqn:noise_model_Z_domain}, becomes
 \begin{align}
\alpha = \sum_{i=1}^N \frac{A^2}{\sigma^2}\frac{1}{1-b_{1}(i)^2}.
\end{align}
Fr our set of $b_i(1)$'s we get a growth rate of $\alpha \simeq 3.0489$ and a limiting false alarm probability of $\mathcal{Q}\left(\dfrac{\sqrt{\alpha}}{{2}} \right)\approx  0.1913$.

 These predictions can be confirmed by the numerical simulations, depicted in Figure~\ref{fig:finite_case}, which gives the empirical probability of false alarm of some agents. 
 \begin{figure}[h!]
	\centering
	\includegraphics[width=7.5cm]{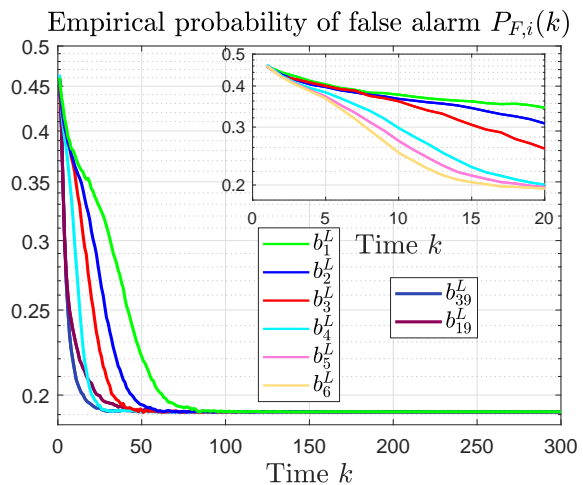}
	\caption{Empirical probability of false alarm  for some agents. All plots are averages over $10^5$ Monte Carlo trials. The coefficients $b_1^L \geq b_2^L \cdots \geq b_N^L$ denote the absolute values of $\{ b_i(1) \colon 1 \leq i \leq N\}$, sorted in non-decreasing order (thus, $b_1^L = \max_{1 \leq i \leq N} | b_i(1)|$ and $b_N^L = \min_{1 \leq i \leq N} | b_i(1) |$). The six agents with the slowest convergence to limit~\eqref{eqn:conf} are the agents with the largest absolute value $| b_i(1) |$.}
	\label{fig:finite_case}
\end{figure}
 (At time $k = 300$, the standard deviation among $\left\{ \mathbb{P}_{F,i}(k) \colon 1 \leq i \leq N \right\}$  is approximately $10^{-16}$, and the deviation between the mean rate $\mathbb{P}_{F,i}(k)$ and limit~\eqref{eqn:conf}  is approximately $10^{-4}$.)

\section{Conclusion}
\label{sec:conclusion}
This paper considers  a distributed detection setup in which both the signal and noise are allowed to follow generic ARMA models, a flexibility that allows to imprint almost any desired temporal correlation in the noise signal by approximating $\mathcal{Z}$ transforms.

To handle this ARMA setup, we suggested an extension of the running consensus detector (RCD) and showed that it can be implemented efficiently at each agent, 
with computational requirements independent of time and proportional to the sparsity of the communication network.

As our theoretical analysis shows, the ARMA setup can lead to a novel regime in the asymptotic behavior of the error probabilities. Specifically, in simple standard (non-ARMA) setups, the probabilities of error of the local detectors typically decay exponentially fast to zero; the ARMA setting, however, can trigger an additional regime in which the probabilities of error converge to a positive error floor. We fully characterize these two regimes by deriving closed-form expressions for both the exponential decay rates and the positive error floors, as a function of the underlying ARMA models. Our analysis also reveals two features not present in standard detection literature. Firstly, the threshold constant used in RCD greatly influences the error floors attained in the additional regime. Secondly, the ARMA setup can lead to a split of agents into strongly informative ones and weakly informative ones, the latter being agents whose measurements can be safely discarded (or not made at all) if so desired, for these measurements do not impact the asymptotic performance of RCD.  Our theoretical findings are confirmed by a set of numerical simulations.
\bibliographystyle{IEEEtran}

\bibliography{IEEEabrv,M335}
\appendix 


We start with the following elementary result, which will be used in the proof of Lemma~\ref{lemma:aux1}:
if $( a(k) )$ is a sequence with non-negative terms (that is, $a(k) \geq 0$ for all~$k$) that satisfies
\begin{align}
    \label{eqn:auxr}
    a(k) \leq r\, a(k-1) + b(k), \quad k \geq 0,
\end{align}
(with $a(-1) = 0$) where $| r | < 1$ and $b(k) \rightarrow 0$, then $a(k) \rightarrow 0$. This elementary result can be shown as follows: because $( b(k) )$ is a convergent sequence, it follows that it is bounded, say, $| b(k) | \leq B$ for all $k \geq 0$; this implies that the sequence $( a (k) )$ is itself bounded: $a(k) \leq A$ for all $k \geq 0$ (for example, it can be shown by induction that $A = B / ( 1 - | r |)$ works); as a consequence, we have $L = \lim\sup_k  a(k) $ to be finite (recall that $\lim\sup_k a(k) = \inf_k A(k)$, where $A(k) = \sup\{ a(k), a(k+1), a(k+2), \ldots \}$); now, from~\eqref{eqn:auxr}, we have $a(k) \leq | r | a(k-1)  + | b(k) |$, and it follows that $L \leq | r | L$ by  the sub-additivity of $\lim\sup$ ($\lim\sup_k \left( \alpha(k) + \beta(k) \right) \leq \lim\sup_k \alpha(k) + \lim\sup_k \beta(k)$) and its shift-invariance ($\lim\sup_k \alpha(k) = \lim\sup_k \alpha(k-1)$), together with $\lim\sup_k b(k) = 0$; finally, intersecting $L \leq | r | L$ with $| r | < 1$ and $L \geq 0$ gives $L = 0$, which is equivalent to $a(k) \rightarrow 0$.

We now turn to the proof of Lemma~\ref{lemma:aux1}. The eigenvalue decomposition~\eqref{eqn:gevd} provides a natural orthonormal basis for $\mathbf{R}^n$, namely, the columns of the the orthonormal matrix $\begin{bmatrix} \frac{1}{ \sqrt{n}} \pmb{1}_n & \mathcal{S} \end{bmatrix}$, thus allowing a generic vector $x \in \mathbf{R}^n$ to be written (uniquely) as $$x =  \frac{1}{\sqrt{n}} \pmb{1}_n \overline x + {\mathcal S} \widetilde x,$$
where $\overline x = \pmb{1}_n^T x / \sqrt{n}$ and $\widetilde x = {\mathcal S}^T x$. In this system of coordinates, showing~\eqref{eqn:lemmaconc} becomes equivalent to showing
\begin{align}
    \frac{\overline v(k)}{f(k)} & \rightarrow \frac{\theta}{\sqrt{n}} \label{eqn:equiv_a} \\
    \frac{\widetilde v(k)}{f(k)} & \rightarrow 0. \label{eqn:equiv_b}
\end{align}

\textbf{Proof of~\eqref{eqn:equiv_a}.} Multiplying both sides of the recursion~\eqref{eqn:template} on the left by $\pmb{1}_n^T \sqrt{n}$ and using $\pmb{1}_n^T \mathcal{W} = \pmb{1}_n^T$ gives the recursion \begin{align} \overline v(k) = \overline v(k-1) + \pmb{1}_n^T u(k) / \sqrt{n}. \label{eqn:ad} \end{align}
Recalling that $\overline v(-1) = v(-1)=0$, we have from~\eqref{eqn:ad} that \begin{align} \frac{\overline v(k)}{f(k)} = \frac{1}{\sqrt{n}} \frac{1}{f(k)} \sum_{m = 0}^k \pmb{1}_n^T u(m). \label{eqn:ad2} \end{align}
Use now the growing rate assumption~\eqref{eqn:ggr} in~\eqref{eqn:ad2} to obtain~\eqref{eqn:equiv_a}.

\textbf{Proof of~\eqref{eqn:equiv_b}.} Multiplying both sides of the recursion~\eqref{eqn:template} on the left by $\mathcal{S}^T$ and using $\mathcal{S}^T \mathcal{W} = \mathcal{D} \mathcal{S}^T$ gives the recursion
\begin{align} \widetilde v(k) = \mathcal{D}\, \widetilde v(k-1) + \mathcal{S}^T u(k), \label{eqn:ad3} \end{align}
or, in components,
\begin{align} \widetilde v_j(k) = d_j\, \widetilde v_j(k-1) + s_j^T u(k), \quad 1 \leq j \leq n-1, \label{eqn:ad4} \end{align}
where $\widetilde v_j(k)$  denotes the $j$th component of the vector $\widetilde v(k)$, $d_j$ is the $j$th diagonal entry of the diagonal matrix $\mathcal{D}$, and $s_j$ is the $j$th column of $\mathcal{S}$.

Dividing~\eqref{eqn:ad4} by $f(k)$ gives
\begin{align} \frac{\widetilde v_j(k)}{f(k)} = d_j\, \frac{f(k-1)}{f(k)} \frac{\widetilde v_j(k-1)}{f(k-1)} + \frac{s_j^T u(k)}{f(k)},  \label{eqn:ad5} \end{align}
    which implies
    \begin{align}
        \left| \frac{\widetilde v_j(k)}{f(k)} \right| & \leq | d_j | \left| \frac{\widetilde v_j(k-1)}{f(k-1)} \right| + \left| \frac{s_j^T u(k)}{f(k)} \right| \\ & \leq  | d_j | \left| \frac{\widetilde v_j(k-1)}{f(k-1)} \right| + \left\| s_j \right\|_\infty \frac{\pmb{1}_n^T u(k)}{f(k)}, \label{eqn:ad6}
    \end{align}
where we used $f(k) > 0$ and $f(k-1) \leq f(k)$ (which holds either in the case $f(k) = 1$ for all $k$, or $f(k) = k+1$ for all $k$) to obtain the first inequality, and we used $a^T b \leq \left\| a \right\|_\infty \left\| b \right\|
_1$ for generic vectors $a = (a_1, \ldots, a_n)$ and $b = (b_1, \ldots, b_n)$ (where $\left\| a \right\|_\infty = \max\{ | a_1 |, \ldots, | a_n | \}$ and $\left\| b \right\|_1 = | b_1 | + \cdots + | b _n |$) to obtain the second inequality (note that $\left\| u(k) \right\|_1 = \pmb{1}_n^T u(k)$ because each $u(k)$ is assumed to be a non-negative vector). Now, note that \begin{align} \frac{\pmb{1}_n^T u(k)}{f(k)} \rightarrow 0. \label{eqn:sr} \end{align}
Indeed, if $f(k) = 1$ for all $k$, then~\eqref{eqn:ggr} shows that the series $\sum_{k \geq 0} \pmb{1}_n^T u(k)$ is convergent (to $\theta$); this implies $\pmb{1}_n^T u(k) \rightarrow 0$, which gives~\eqref{eqn:sr}. If, instead, $f(k) = k+1$ for all $k$, then 
\begin{align}
    \frac{\pmb{1}_n^T u(k)}{k+1} & = \underbrace{\frac{1}{k+1} \sum_{m = 0}^k \pmb{1}_n^T u(m)}_{\alpha(k)} - \frac{k}{k+1} \underbrace{\left( \frac{1}{k} \sum_{m = 0}^{k-1} \pmb{1}_n^T u(m) \right)}_{\beta(k)}. \label{eqn:ad5}
\end{align}
Because from~\eqref{eqn:ggr} we have $\alpha(k) \rightarrow \theta$ and $\beta(k) \rightarrow \theta$, taking the limit $k \rightarrow \infty$ in~\eqref{eqn:ad5} gives~\eqref{eqn:sr}.

Finally, using~\eqref{eqn:sr} and the assumption $| d_j | < 1$, we see that~\eqref{eqn:ad6} fits the template~\eqref{eqn:auxr}, from which we conclude $\widetilde v_j(k) / f(k) \rightarrow 0$ for $1 \leq j \leq n-1$; that is, we obtain~\eqref{eqn:equiv_b}.

\subsection{Proof of Lemma~\ref{lemma:aux2}} 

Our proof uses the well-known bounds
	\begin{align}
	\frac{t}{1+t^2} \exp\left(-\frac{t^2}{2}\right) \leq \sqrt{2\pi}\, \mathcal{Q}(t) \leq \frac{1}{t} \exp\left(-\frac{t^2}{2}\right),
	\label{eqn:bounds_Q_function}
	\end{align} 
 valid for $t > 0$, or, equivalently,
 \begin{align}
\frac{t^2}{2} + \alpha(t) \leq - \log \mathcal{Q}(t) \leq \frac{t^2}{2} + \beta(t), \label{eqn:equivQ}
 \end{align}
 where $\alpha(t) = \log\left( \sqrt{2 \pi} t \right)$ and $\beta(t) = \log\left( \sqrt{2 \pi} \left( t + 1/t \right) \right)$ for $t > 0$.

 The inequality in~\eqref{eqn:equivQ} gives
\begin{align}
\frac{u(k)^2}{2} + \frac{\alpha\left( u(k) \sqrt{v(k)} \right)}{v(k)} \leq - \frac{\log \mathcal{Q}\left( u(k) \sqrt{v(k)} \right)}{v(k)}, \label{eqn:d1}
\end{align}
whose left-hand side verifies
\begin{align}
\frac{u(k)^2}{2} + \frac{\alpha\left( u(k) \sqrt{v(k)} \right)}{v(k)} \rightarrow \frac{\theta^2}{2} \label{eqn:d2}
\end{align}
 because $u(k) \rightarrow \theta$ and $v(k) \rightarrow \infty$.

The second inequality in~\eqref{eqn:equivQ} gives
\begin{align}
 - \frac{\log \mathcal{Q}\left( u(k) \sqrt{v(k)} \right)}{v(k)} \leq 
\frac{u(k)^2}{2} + \frac{\beta\left( u(k) \sqrt{v(k)} \right)}{v(k)}.  \label{eqn:d3}
\end{align}
Now, writing $\beta(t) = \log\left( \sqrt{2 \pi} t \left( 1 + 1/t^2 \right) \right)$ and using the general inequality $\log( 1 + r ) \leq r$ for $r > 0$, we get
$$\log\left( \sqrt{2 \pi} t \right) \leq \beta(t) \leq \log\left( \sqrt{2 \pi} t \right) + 1/t^2,$$
from which it easily follows that $
{\beta\left( u(k) \sqrt{v(k)} \right)}/{v(k)} \rightarrow 0.$
Thus, the right-hand side of~\eqref{eqn:d3} satisfies
 \begin{align}
\frac{u(k)^2}{2} + \frac{\beta\left( u(k) \sqrt{v(k)} \right)}{v(k)} \rightarrow \frac{\theta^2}{2}. \label{eqn:d4}
\end{align}
Results~\eqref{eqn:d1},~\eqref{eqn:d2},~\eqref{eqn:d3}, and~\eqref{eqn:d4} imply~\eqref{eqn:convQ}.

\end{document}